\documentclass[12pt,tightenlines,eqsecnum,floatfix,superscriptaddress,showpacs,amssymb,amsmath,amsfonts,aps,altaffilletter,nofootinbib,longbibliography,showpacs]{revtex4-2}

\usepackage{color}
\usepackage{hyperref}
\usepackage{amsmath}
\usepackage{amsthm}
\usepackage{multirow}
\usepackage{graphicx}
\usepackage{tikz}
\usepackage{caption}
\usepackage{subcaption}
\usepackage{placeins}
\usepackage{mathrsfs}
\usepackage[colorinlistoftodos]{todonotes}
\usepackage{epstopdf}
\usepackage{caption}
\DeclareSymbolFontAlphabet{\mathrsfs}{rsfs}
\newcommand{\scrM}{\mathrsfs{M}}
\newcommand{\scrR}{\mathrsfs{R}}
\newcommand{\DHS}{\mathrsfs{H}}
\newcommand{\E}{\mathrsfs{E}}
\newcommand{\Lie}{\mathrsfs{L}}
\newcommand{\Surs}{\ensuremath{\mathcal{S}}}
\def\uk{\underline{k}}

\newtheorem{defn}{Definition}
\newtheorem{lemma}{Lemma}
\newtheorem{prop}{Proposition}

\def\be{\begin{equation}}
\def\ee{\end{equation}}
\def\ba{\begin{eqnarray}}
\def\ea{\end{eqnarray}}
\newcommand{\pb}[1]{\hbox{\lower0.5ex\hbox{${}_{\leftarrow}$}}\kern-1.9ex{#1}}

\def\rmd{\mathrm{d}}
\def\scri{\mathscr{I}}
\def\scrip{\scri^{+}}
\def\t{\tilde}
\def\R{\mathcal{R}}
\def\h{\hat}
\def\f{\frac}

\def\qo{\mathring{q}}
\def\qot{\mathring{\tilde{q}}}
\def\psio{\mathring{\psi}}
\def\Ko{\mathring{K}}
\def\Lo{\mathring{L}}
\def\eo{\mathring\epsilon}
\def\eot{\mathring{\t\epsilon}}
\def\hr{\hat{r}}
\def\={\hat{=}}
\def\Dot{\mathring{\tilde{D}}}

\begin{document}

\title[]{Dynamical Horizon Segments and Spacetime Isometries}

\author{Abhay Ashtekar}
\affiliation{Institute for Gravitation and the Cosmos \& Physics Department, Pennsylvania State University, University Park, PA 16802, U.S.A.}
\email{ashtekar.gravity@gmail.com}

\author{Badri Krishnan}
\affiliation{Institute for Mathematics, Astrophysics and Particle Physics, Radboud University, Heyendaalseweg 135, 6525 AJ Nijmegen, The Netherlands}
\affiliation{Albert-Einstein-Institut, Max-Planck-Institut f\"ur Gravitationsphysik, Callinstra{\ss}e 38, 30167 Hannover, Germany}
\affiliation{Leibniz Universit\"at Hannover, 30167 Hannover, Germany}
\email{badri.krishnan@ru.nl}

\begin{abstract}

Given a space-time $(\mathrsfs{M}, g_{ab})$ admitting a dynamical horizon segment $\DHS$, we show that there are stringent constraints on the Killing fields $\xi^a$ that $g_{ab}$ can admit in a neighborhood of $\DHS$: Generically, $\xi^a$ can only be a rotational Killing field which, furthermore, leaves each marginally trapped 2-sphere cross-section $\Surs$ of $\DHS$ invariant. Finally, if $\xi^a$ happens to be hypersurface orthogonal near $\DHS$, then, not only the angular momentum but also \emph{all} spin multipoles vanish on every $\Surs$; the \emph{entire} spin structure of these DHSs is indistinguishable from that of spherically symmetric DHSs!
  
\end{abstract}

\maketitle

\section{Introduction}
\label{s1}

By now it is well-recognized that the use of event horizons (EHs) $\E$ to characterize black holes (BHs) is untenable in \emph{dynamical} situations (see, e.g., \cite{Ashtekar:2004cn,Visser_2014}). The most important limitation comes from the fact that EHs are teleological. In particular, they can form and grow in flat regions of space-time where nothing at all is happening, \emph{in anticipation} of a collapse that may take place a million years into the future (see, e.g., \cite{Ashtekar:2003hk,kehle2024extremalblackholeformation})! For, to determine if an EH forms in a given stellar collapse, or to locate it in a space-time representing a binary BH merger, one needs to know space-time geometry all the way to the infinite future. Consequently, in the course of a simulation, EHs cannot be used to track the progenitors, nor to locate the merger; they can only be found at the end of the simulation, as an after thought. Also,  their definition makes an essential use of the future boundary $\scrip$ of space-time, which furthermore has to be complete \cite{Geroch:1978ub}. This feature makes them unsuitable in the exploration of conceptual issues where it is simply inappropriate to make \emph{a priori} global assumptions on whether the future evolution of the space-time of interest will extend  all the way to include a complete $\scrip$. This is particularly true in the investigation of cosmic censorship in classical general relativity (GR) \cite{Andersson:2007fh}, and in explorations of the final state of a BH in quantum gravity, at the end of the evaporation process \cite{Ashtekar_2020,afshordi2024blackholesinside2024}. 

Because of these limitations, a teleology-free, quasi-local notion of \emph{dynamical horizons} was introduced two decades ago \cite{Ashtekar:2002ag,Ashtekar:2003hk,Ashtekar:2004cn,Booth:2005qc}. Subsequently it  has been analyzed extensively from a mathematical relativity perspective (see, in particular, \cite{Booth:2003ji,Senovilla:2003tw,Andersson:2005gq,Ashtekar:2005ez,Andersson:2007fh,Andersson:2008up,Jaramillo:2011pg,Ashtekar:2013qta,Booth_2020}), and used widely in numerical relativity (see, in particular, \cite{Jaramillo:2011re,Jaramillo:2011rf,Mourier:2020mwa,Booth:2021sow,Pook-Kolb:2021gsh,Pook-Kolb:2021jpd,Chen_2022,Khera:2023oyf}) and in the investigation of the quantum evaporation of BHs (see, e.g., \cite{Ashtekar:2005cj,Hayward_2006,Ashtekar_2008,Ashtekar_2011a,Ashtekar_2011b,Frolov_2014,frolov2014blackholesexist,Bianchi_2018,ashtekar2025blackholeevaporationloop}). In particular, these investigations have led to detailed expressions relating the change in the area of dynamical horizons to the flux of energy falling across them, furnished multipole moments to track the evolution of their shape and spin structure in an invariant manner, and provided a great deal of analytical and numerical control on the formation and growth of dynamical horizons. Thanks to these new inputs, a sharper notion of \emph{dynamical horizon segments} (DHSs) has now emerged \cite{ashtekar2025quasilocalblackholehorizons}. 

In this note we will analyze the interplay between the geometry of DHSs $\DHS$ and isometries that the space-time metric $g_{ab}$ can admit in a neighborhood of $\DHS$. We will find that the very existence of a DHS severely constrains possible isometries and conversely the presence of isometries constrains the geometry of DHSs in a surprisingly strong fashion.

This note is organized as follows. In section \ref{s2} we recall the basic notions and results from the literature. In section \ref{s3} we focus on Killing fields that are causal in a neighborhood of a DHS and in section \ref{s4} those that are spacelike. The arguments involved are different from those used in the literature summarized in section \ref{s2.3} and the results are stronger in two respects: they are free of additional assumptions on DHSs that were made in earlier discussions, and they rule out possibilities that were not contemplated there. In section \ref{s5} we discuss the complementary issue: restrictions on the geometry of DHSs imposed by the presence of permissible isometries. We show that if the Killing field is hypersurface orthogonal, the restrictions are surpassingly strong. Section \ref{s6} summarizes the main results and their implications. In Appendix A we discuss restrictions on Killing fields that a smooth metric on a 2-sphere $S^2$ can admit: Each of these Killing fields is necessarily a rotation of a unit round metric on $S^2$. This result is used in sections \ref{s4} and \ref{s5}.

Our conventions are the following. We work with 4-dimensional space-times $\scrM$ with metrics $g_{ab}$ of signature -,+,+,+. The torsion-free derivative operator compatible with the metric is denotes by $\nabla$ and the curvature tensors are defined by $R_{abc}{}^d\, V_d := 2 \nabla_{[a} \nabla_{b]}\, V_c; \, R_{ac} :=  R_{abc}{}^{b}$ and $R= g^{ab}\, R_{ab}$. For simplicity we will assume that $\scrM$ and all fields under consideration are smooth.

\section{Preliminaries}
\label{s2}

This section is divided into three parts. In the first we recall the basic notions and the associated structures which will be used in the subsequent discussion. In the second, we summarize the notion of multipoles that characterize the shape and spin structure of DHSs in an invariant fashion. This notion will be used in section \ref{s5}. In the third we collect the known results on restrictions on space-time isometries imposed by the presence of a DHS. These results are extended in sections \ref{s3} and \ref{s4}.

\subsection{Basic notions}
\label{s2.1}
Let $\Surs$ be a closed spacelike 2-dimensional sub-manifold in a
4-dimensional space-time $(\scrM, g_{ab})$.  Let $q_{ab}$ be the
intrinsic metric on $\Surs$, $\epsilon_{ab}$ the area 2-form, and
let $\nabla\!_{a}$ be the spacetime derivative operator compatible with
$g_{ab}$.  Since $\Surs$ is spacelike, it admits two future directed
null normal vector fields $k^a$ and $\uk^a$, defined up to rescalings
by positive definite functions.  Without loss of generality we can
rescale them so that $k^a\uk_a =-1$. (All our conclusions will be 
insensitive to the remaining rescaling freedom $k^a \to f k^a$ and 
$\uk_a \to f^{-1}\,\uk^a$, for some positive function $f$ on $\Surs$.) The expansions $\Theta_{(k)}$ and $\Theta_{(\uk)}$ of $k^a$ and $\uk^a$, are defined as
\begin{equation}
  \Theta_{(k)} = q^{ab}\nabla_a k_b\,,\quad  \Theta_{(\uk)} = q^{ab}\nabla_a \uk_b\, , 
\end{equation}
where the inverse metric $q^{ab}$ on $\Surs$ is used to project space-time fields onto $\Surs$. We shall also need the shear $\sigma^{(k)}_{ab}$ of $k^a$ which is a symmetric trace-free tensor that, together with $\Theta_{(k)}$, determines the projection of the symmetrized derivative of $k^a$ to $\Surs$:
\begin{equation}
  \sigma_{ab}^{(k)} = \left(q_a{}^c q_b{}^d\, -\, \frac{1}{2}q_{ab} q^{cd}\right) \nabla_{(c} k_{d)} \,.
\end{equation}
\smallskip

\begin{defn}
 A closed, 2-dimensional spacelike submanifold $\Surs$\, of\, $(\scrM, g_{ab})$ is said to be a \emph{marginally trapped surface} (MTS) if the expansion of one of the null normals, say $k^a$, vanishes, i.e. if\, $\Theta_{(k)} =0$.  \end{defn}

Note that no restriction is placed on $\Theta_{(\uk)}$. Although it is not necessary for our considerations, in asymptotically flat space-times, $k^a$ is often taken to be the `outward pointing' null normal (and denoted by $\ell^a$), while $\uk^a$ is taken to be `inward pointing' (and denoted by $n^a$). 

World-tubes of MTSs provide a notion of Dynamical Horizon segments \cite{ashtekar2025quasilocalblackholehorizons} that is central to out discussion:
\begin{defn} A connected 3-dimensional sub-manifold $\DHS$ of a spacetime $(\scrM,g_{ab})$ which is nowhere null is said to be a Dynamical Horizon Segment (DHS) if
  \begin{enumerate}
  \item $\DHS$ 
  is foliated by MTSs $\Surs$, each with topology $S^2$. {\rm The future pointing null normal 
    to the MTSs with vanishing expansion will be denoted by $k^a$; {\smash{$\Theta_{(k)} = 0$.}}
        }
  \item If $\uk^a$ is the other future pointing null normal to the MTSs $\Surs$, then $\Theta_{(\uk)}$ is nowhere vanishing. 
      \item The quantity
    \begin{equation}
      \E = \sigma_{ab}^{(k)}\sigma^{ab}_{(k)} + R_{ab}k^ak^b\,,
    \end{equation}
    does not vanish identically on any MTS $\Surs$. {\rm Here $R_{ab}$ is the space-time Ricci tensor.}
  \end{enumerate}
\end{defn}
\noindent Note that since $\DHS$ is connected, it is either spacelike or timelike everywhere, and $\Theta_{(\uk)}$ is either positive or negative everywhere on $\DHS$. In applications, DHSs arise both as manifolds with and without boundary. For details, see \cite{ashtekar2025quasilocalblackholehorizons}. We will assume that $k^a$ and $\uk^a$ are normalized such that $k^a \uk_a =-1$.
\vskip0.1cm

{\bf Remarks:} 
\vskip0.1cm

1. If $k^a$ were null, then the 3-manifold $\DHS$ would have been null, i.e., would have been a non-expanding horizon (NEH). The NEHs are well suited to describe equilibrium situations in which there is no energy flux into the BH \cite{Ashtekar:2000hw,Ashtekar:1998sp,Ashtekar:2004cn}. In particular the event horizons of \emph{stationary} black holes are NEHs. The DHSs discussed in this paper represent the complementary situation in which the horizon geometry is time-dependent in response to infalling matter or gravitational waves.
\vskip0.1cm

2. \emph{Dynamical horizons} introduced in \cite{Ashtekar:2002ag,Ashtekar:2003hk} were more restrictive than the DHSs of {Definition 2} in that they were required to be \emph{spacelike} sub-manifolds, and furthermore, the expansion $\Theta_{(\uk)}$ of the other null normal was required to be \emph{strictly negative}. These restrictions were motivated by the properties of horizons that arise in BH mergers: The DHSs of the progenitors as well as those of the remnants satisfy these additional conditions almost everywhere in their long history, except very close to the merger. However, timelike DHSs do arise in some interesting situations such as in the Oppenheimer-Snyder collapse in classical general relativity, and also during the BH evaporation process in quantum gravity. These situations are encompassed by DHSs but not by dynamical horizons.% 
\footnote{Also, even when the DHS is spacelike, $\Theta_{(\uk)}$ can be positive if the DHS separates an anti-trapped region from an untrapped region; this occurs for example in a process in which a white hole evaporates by emitting null fluid carrying positive energy. For a discussion of various situations in which one has DHSs that are not dynamical horizons, see section III.A of \cite{ashtekar2025quasilocalblackholehorizons}.}
On the other hand, condition 3 of {Definition 2} does not feature in the definition of dynamical horizons. This condition serves to eliminate exceptional situations \cite{Senovilla:2003tw} in which dynamical horizons exist although the space-time does not have a BH connotation because there are no trapped regions. (These space-times are exceptional in that all curvature invariants vanish on them.) Thanks to condition 3, such space-times do not admit DHSs. 

In this paper we work with DHSs, \emph{not} dynamical horizons. An immediate consequence of {Definition 2} is:

\begin{lemma}
  \label{lemma1}
On any DHS $\DHS$, the area of the MTSs is a monotonic function. 
\end{lemma}
\begin{proof}
 Let $V^a$ be a vector field on $\DHS$ that is everywhere
orthogonal to each MTS, and is such that the flow generated by $V^a$ preserves the foliation. Then $V^a = f_1 k^a + f_2 \uk^a$ for some nowhere vanishing functions $f_1$ and $f_2$. If $\DHS$ is spacelike (respectively, timelike) then $V^a$ is spacelike (respectively, timelike). Let $\epsilon_{ab}$ be the volume 2-form on each $\Surs_\lambda$. Then
\begin{equation} \label{epsilondot}
\Lie_{V}\epsilon_{ab} = f_2\,\Theta_{(\uk)}\,\epsilon_{ab}\,.
\end{equation}
Let us denote leaves of the foliation by $\Surs_\lambda$ using an affine parameter $\lambda$ of $V^a$ that is constant on each MTS. Since $f_2\Theta_{(\uk)}$ is nowhere vanishing, by integrating (\ref{epsilondot}) on $\Surs_\lambda$, we conclude that $dA_\lambda/d\lambda$ is either strictly positive or strictly negative.
\end{proof}

Monotonicity of area will play a key role in sections \ref{s3} and \ref{s4}.

\subsection{Multipoles of DHSs}
\label{s2.2}

Multipole moments are a set of numbers that provide an invariant characterization of the shape and spin structure of the MTSs $\Surs$ of a DHS. These numbers change as one passes from one MTS to another and this time dependence enables one to study the evolution of the DHS geometry. For example, numerical simulations show that while the MTSs on the common DHS are highly distorted and their spin structure has a strong time dependence immediately after a common DHS forms in a BH merger, this DHS settles down to a Kerr isolated horizon segment quite quickly. However, these descriptions generally depend on coordinate choices. Using multipole moments, one can provide an invariant description of the passage to the final equilibrium state \cite{Ashtekar:2013qta,Chen_2022}. 

Because each MTS $\Surs$ has 2-sphere topology, the metric $\t{q}_{ab}$ on it is conformally related to the \emph{round} 2-sphere metric. The conformal factor, that carries information about distortions, is completely determined by the scalar curvature $\t\R$ of $\t{q}_{ab}$. Therefore, the gauge invariant information about shape of $\Surs$ is encoded in  $\t\R$. Information about its spin --i.e., intrinsic angular momentum-- is encoded in (the exterior derivative of) the so-called `rotational 1-form' \cite{Ashtekar:2013qta}
\be \label{tomega}\t\omega_a := -\t{q}_a{}^b\, \uk_c \nabla_b k^c = \t{q}_a{}^b\, K_{bc}\, \h{r}^c\, , \ee 
where $K_{ab}$ is the extrinsic curvature of $\DHS$ in $(\scrM, g_{ab})$, $\h{r}^c$ is the unit normal to $\Surs$ within $\DHS$, and we have chosen a rescaling of the null normals so that $k^a \h{r}_a = \f{1}{\sqrt{2}}$. To illustrate this encoding of angular momentum, let us consider a partial Cauchy surface that extends to spatial infinity and has an inner boundary $\Surs$, and a diffemorphism generated by a vector field $\varphi^a$ which is a rotational Killing field on $\Surs$. Then, in the Arnowitt-Deser-Misner phase space, the horizon contribution to the Hamiltonian generating diffeomorphisms along $\varphi^a$ --or the value of the angular momentum charge at $\Surs$--  is given by $J_{\varphi} [\Surs] = - \f{1}{8\pi G}\, \oint_{\Surs} \t\omega_a \varphi^a\, \rmd^2 \t{V}$. (Using the divergence-free property of $\varphi^a$, it follows that $J_{\varphi} [\Surs]$ depends only on the exterior derivative of $\t\omega_a$.) Axisymmetric DHSs arise in an axisymmetric collapse as well as in a head on collision of black holes. 

With this application in mind, we will now summarize the procedure to define multipole moments of  an axisymmetric DHS $\DHS$ \cite{Ashtekar:2004gp,Schnetter:2006yt}, i.e. of a DHS on which the intrinsic metric $\t{q}_{ab}$ on each MTS $\Surs$ is axisymmetric.%
\footnote{Multipole moments are also defined on general DHSs without reference to axisymmetry \cite{Ashtekar:2013qta} but the procedure is much more involved. In the axisymmetric case the general procedure coincides with the simpler one discussed here.}
For reasons discussed in sections \ref{s2.3} and \ref{s4}, we can assume that the rotational Killing field $\varphi^a$ is tangential to every MTS $\Surs$ of $\DHS$. Then, using $\varphi^a$ and the intrinsic metric $\t{q}_{ab}$ on $\Surs$, one can introduce invariantly defined coordinates  $\zeta \in [-1,1]$ and $\varphi \in [0, 2\pi)$ as follows:  $\zeta$ is given by $\tilde{D}_a\zeta = R^{-2}\, \tilde{\epsilon}_{ba}\varphi^b$, where $R$ is the area radius of $S$ and $\varphi$ is an affine parameter of $\varphi^a$ such that $\t{q}^{ab} \t{D}_a\varphi \t{D}_b \zeta =0$. ($\zeta$ is the analog of the function $\cos\theta$ in usual spherical coordinates.) The freedom of adding a constant to $\zeta$ is removed by requiring $\oint_S\zeta\,\,\tilde{\epsilon}_{ab} = 0$,  but there remains a freedom to shift $\varphi$ by a rigid translation. In axisymmetric space-times, the DHS multipoles are insensitive to this freedom. 

In terms of these coordinates, the 2-metric $\tilde{q}_{ab}$ takes the canonical form \cite{Ashtekar:2004gp}:
\begin{equation}
  \label{eq:canonical-2metric}
  \t{q}_{ab} = R^2\,(f^{-1}\tilde{D}_a\zeta\tilde{D}_b\zeta + f\tilde{D}_a\varphi\tilde{D}_b\varphi)\,,
\end{equation}
where $f\,=\, R^{-2}\, \varphi_a \varphi^a$. It is easy to check that $\t{q}_{ab}$ is a round 2-sphere metric if  $f =1 - \zeta^2$. Now, by inspection, the area-element corresponding the generic axisymmetric $\t{q}_{ab}$ of Eq.~(\ref{eq:canonical-2metric}) is independent of $f$ and is therefore the same as that of the round 2-sphere. Multipole moments are defined using the spherical harmonics $Y_{\ell, m}$ of the round 2-sphere metric as weight functions.  The inner-products between the $Y_{\ell, m}$ computed using the physical 2-metric $\t{q}_{ab}$ are the same as for the standard spherical harmonics on a round 2-sphere, since the two area elements are the same. With this structure on hand one can define `shape' multipoles $\texttt{I}_\ell$ and `spin' multipoles $\texttt{S}_\ell$ \cite{Ashtekar:2013qta}:
\ba  \label{dimensionless2} \texttt{I}_{\ell,\,m} [\Surs]\, &:=& \textstyle{\f{1}{4}}\, \oint_{\Surs} \big[\,\t\R]\,Y_{\ell,\,m} (\theta,\varphi) \\
\texttt{S}_{\ell,\, m} [\Surs] &:=& \textstyle{\f{1}{4}}\, \oint_{\Surs} \big[2 \,\t\epsilon^{ab} \t{D}_a \t\omega_b\,\big]\, Y_{\ell,\,m} (\theta,\varphi) \, \rmd^2 \t{V}\, . \ea
where, $\theta$ is defined via $\zeta= cos\theta$ and, as usual, $\ell \in [0, \infty)$ and $m\in [-\ell, \ell]$. The $\texttt{I}_{\ell,m}$ and  $\texttt{S}_{\ell,m}$ are dimensionless and called  \emph{geometrical} multipoles. To define the mass and angular momentum multipoles  $\texttt{M}_{\ell,m}, \texttt{J}_{\ell,m}$, one introduces an `effective mass surface density' and an `effective spin-current density' on MTSs  using an analogy with electrodynamics. The final expressions of $\texttt{M}_{\ell,m}[\Surs],\, \texttt{J}_{\ell,m}[\Surs]$ are just rescalings of  $\texttt{I}_{\ell,\,m}[\Surs],\, \texttt{S}_{\ell,\, m} [\Surs]$ with appropriate dimensional factors \cite{Ashtekar:2004gp}. Therefore, in this note we will work with the simpler geometrical multipoles. In the axisymmetric case now under consideration, all multipoles with $m\not=0$ vanish. Consequently, as noted above, the freedom to make rigid shift in $\varphi$ does not affect the multipoles.

These multipoles have a number of physically expected properties:\\
(i) The mass dipole and angular momentum monopole vanish identically on all MTSs.\\
(ii) If fields on an MTS are spherically symmetric,  all multipoles except the mass monopole vanish. If they are reflection-symmetric, (as on isolated horizons of the Kerr family), then  all $\texttt{M}_\ell$ vanish for odd $\ell$ and all $\texttt{J}_\ell$ vanish for even $\ell$.\\
(iii) As pointed out above, in the standard Hamiltonian framework for partial spacelike surfaces, angular momentum $J_{(\varphi)}[\Surs]$ arises as the `horizon surface charge' defined by the canonical transformation generated by the symmetry vector field $\varphi^a$. That property is carried over to all higher angular momentum multipoles as well. The $\texttt{S}_{\ell,\, 0} [\Surs]$ are the horizon surface charges corresponding to the canonical transformations generated by vector fields $\t\epsilon^{ab} D_a Y_{\ell, 0}(\zeta)$ (up to overall constants).

\vskip0.1cm

\subsection{Known restrictions on isometries}
\label{s2.3}

In this section we will briefly summarize the results on permissible isometries in neighborhoods of DHSs that have appeared in the literature. 

The first analysis \cite{Mars:2003ud} showed that \emph{strictly stationary spacetimes cannot contain closed trapped or closed marginally trapped surfaces}. Here the term `closed' means compact and without boundary, and the phrase `strictly stationary space-times' refers to space-times that admit a Killing field whose norm is \emph{strictly} negative. As emphasized in that paper, both these conditions are indispensable for the method used in the proof. For example, even Minkowski space admits MTSs that are not closed. Similarly, there are space-times that admit null Killing fields and multitude of MTSs \cite{Senovilla:2003tw}, and of course the Kerr family admits MTSs as well as a  Killing field that is stationary in an asymptotic region but is not \emph{strictly} stationary on the full space-time. From the perspective of dynamical horizons used in this paper, this result implies that a strictly stationary space-time does not admit a DHS, providing support for the intended connotation that DHSs arise only in dynamical situations. However, the result is more general in that it refers only to a 2-surfaces, not to its `time evolution'. In particular, it also implies that strictly stationary space-times cannot admit non-expanding horizons (NEHs) which are null and do not have a dynamical connotation.

The second set of results \cite{Ashtekar:2005ez} arose as collorolies of general theorems on uniqueness of dynamical horizons a la \cite{Ashtekar:2002ag,Ashtekar:2003hk}, rather than the DHSs considered here. But the first of these results holds also for DHSs because it did not make use of the more restrictive properties of dynamical horizons. This uniqueness theorem refers to the foliation of a DHS by the MTSs: It shows that the foliation is unique. A corollary of this uniqueness is that if the space-time metric admits a Killing field $\xi^a$ that is everywhere tangential to a DHS $\DHS$, then it must be tangential to each MTS $\Surs$ in $\DHS$. 

However, the subsequent constraints on permissible isometries made a crucial use of a second theorem that holds only on \emph{regular} dynamical horizons $H$; these are spacelike, with $\Theta_{(\uk)} <0$ everywhere, and with a \emph{nowhere} vanishing $\mathcal{E}$. This theorem guarantees that, if the null energy condition holds, then there is no MTS which lies strictly in the past domain of dependence of a regular dynamical horizon $H$ \cite{Ashtekar:2005ez}. A corollary of this result was then noted: since existence of a Killing field  that is everywhere transverse to a MTS $\Surs$ in $H$ would contradict the theorem, such Killing fields cannot exist in a neighborhood of a regular dynamical horizon. 

These results considered Killing fields $\xi^a$ that are either tangential to the given DHS or transverse to MTSs $\Surs$. But they leave open the possibility that $\xi^a$ may be tangential to a closed subset of $\Surs$ and transverse to its complement. This possibility is discussed in \cite{Andersson:2007fh} using of properties of the stability operator of MTSs. (For a concise discussion of the stability operator see, e.g., section IV of \cite{ashtekar2025quasilocalblackholehorizons}). Consider a MTS $S$ in a region of space-time that admits a Killing field $\xi^a$ ($S$ need not lie on a DHS $\DHS$). Denote by $\xi^a_{\perp}$ the projection of $\xi^a$ orthogonal to $S$. Suppose that $S$ is stable along a direction parallel to $\xi^a_{\perp}$. Then $\xi^a$ is either everywhere tangential to $S$ (so that $\xi^a_{\perp}=0$), or everywhere transversal to it \cite{Andersson:2007fh}. 

Taken together, these results imply that if space-time metric admits a Killing field $\xi^a$ in a neighborhood of a DHS $\DHS$, it must be tangential to every MTS $\Surs$ of $\DHS$ provided:\\
(i) the  DHS is spacelike with $\Theta_{(\uk)} <0$;\\ 
(ii) $\E$ is nowhere vanishing;\\
(iii) Null energy condition holds; and,\\
(iv) Its MTSs $\Surs$ are stable.\\
In sections \ref{s3} and \ref{s4} we will analyze the constraints on permissible symmetries without assuming these additional restrictions on the DHS.
\vskip0.1cm

\textbf{Remark:} There are also other results on the interplay between symmetries and the existence of MTSs $S$ that use the stability operator but do not refer to a DHS \cite{Carrasco:2009sa,Booth_2024}. However, their statements refer to partial Cauchy slices with the MTS $S$ as an inner boundary. Therefore they are not directly relevant to the considerations of this note.

\section{Absence of Causal Killing fields\\ in any neighborhood of $\DHS$}
\label{s3}

Since the intrinsic geometry of a DHS $\DHS$ is inherently time dependent, one might intuitively expect that the space-time metric would not admit a causal Killing field in a neighborhood of $\DHS$.
In this section we will show that this expectation is indeed borne out. The section is divided into two parts. In the first we consider the case when the $\xi^a$ is timelike and in the second we let $\xi^a$ be a general causal vector field.

\subsection{timelike Killing vectors}
\label{s3.1}

\begin{prop}
  \label{prop1}
  If a region $\scrR$ in a spacetime $(\scrM,g_{ab})$ admits a timelike Killing vector field, then $\scrR$ cannot contain a  complete MTS $\Surs$ of a DHS. In particular, then, it cannot contain a DHS.
\end{prop}
\begin{proof}
Let $\xi^a$ be the Killing vector field in the region $\scrR$, and suppose that $\scrR$ contains a complete MTS $\Surs$ of a DHS with null normals $(k^a,\uk^a)$.  Let $\lambda$ be an affine parameter along $\xi^a$, and choose $\lambda=0$ on $\Surs$.  Let $\Surs_\lambda$ be the image of $\Surs$ under the isometry generated by $\xi^a$ for all $\lambda\in (-\delta,\delta)$, and denote again by $(k^a,\uk^a)$ the null normals to $\Surs_\lambda$ transported from $\Surs$ using the isometry. Since $\Surs$ is an MTS within a DHS $\DHS$, $\E$ does not vanish identically on $\Surs$. Let us choose a sufficiently small $\delta$ such that $\E$ does not vanish identically on any  $\Surs_\lambda$. Let $\mathfrak{H} = \bigcup_\lambda \Surs_\lambda$  denote the 3-manifold obtained by taking the union of all the $\Surs_\lambda$. Then, $\Theta_{(k)} = 0$ on each $\Surs_\lambda$ and $\Theta_{(\uk)}$ is nowhere vanishing on $\Surs_\lambda$. Thus, $\mathfrak{H}$ is also a DHS. Hence, by Lemma \ref{lemma1}, the area of the 2-spheres $\Surs_\lambda$ must increase or decrease monotonically in $\lambda$. On the other hand, since any $\Surs_\lambda$ is the image of $\Surs$ under an isometry, its area has to equal that of $\Surs$. Thus we would have a contradiction if $\Surs$ were to lie in $\scrR$. In particular, a DHS cannot lie in a stationary region of spacetime.
\end{proof}
\goodbreak

\textbf{Remarks:}
\vskip0.1cm
1. Our argument used only conditions 1 and  2 of {Definition 2}; we did not need the condition on $\mathcal{E}$ in Definition 2 of a DHS. 
\vskip0.1cm

2. As noted above, this result was obtained in \cite{Ashtekar:2005ez} for \emph{regular} dynamical horizons $H$: Thus, $H$ was assumed to be spacelike, with $\Theta_{(\uk)} <0$ everywhere and $\mathcal{E}$ \emph{nowhere} vanishing. Also, the proof was indirect.  The present proof, by contrast, requires fewer assumptions and directly brings out the essential reason that forbids the existence of a timelike Killing field.
\vskip0.1cm

3. After this work was completed, we found out that Mars and Senovilla already had a Theorem for `strictly stationary space-times' \cite{Mars:2003ud}.  This work (discussed in section \ref{s2.3})
implies the result of Proposition \ref{prop1}. Their argument is similar to the one presented above. Their result is more general because they rule out trapped surfaces without reference to a DHS. But, as we show in section \ref{s3.2} for the MTSs that do lie on a DHS, one can extend our result of Proposition \ref{prop1} to include causal Killing vectors. Also, as we discuss in section \ref{s4}, our setting enables one to extend the analysis to include spacelike Killing vectors as well.

\subsection{General causal Killing fields}
\label{s3.2}

In this section we will generalize the previous discussion by considering Killing fields $\xi^a$ that are causal in the region $\scrR$. Thus now $\xi^a$ will be allowed to be \emph{either} timelike \emph{or} null in a neighborhood of a MTS $\Surs$ in $\DHS$; it can change its character even on $\Surs$, so long as it remains causal.  

At any point $p$ of an MTS $\Surs$ in $\DHS$, let $T_p\Surs$ be the tangent space to $\Surs$ and let $\mathbf{e}_{(i)}^a$ ($i=1,2$) be a basis in $T_p\Surs$. Let us first consider the case when $\xi^a$ is nowhere vanishing on an MTS $\Surs$. Since it is causal and $T_P\Surs$ is spacelike, $(\xi^a,\mathbf{e}_{(1)}^a,\mathbf{e}_{(2)}^a)$ span a 3-dimensional vector space, say $W_p$.  Note that $\xi^a$ is not necessarily orthogonal to $\Surs$. But being causal, it is transverse to $\Surs$. Therefore $W_p$ is either a timelike or null 3-flat but its signature can change from one point of $\Surs$ to another. (Note also that the causal nature of $W_p$ is determined not by the causal nature of $\xi^a\mid_p$, but rather by the causal nature of the normal $V^a$ to $T_p\Surs$ within $W_p$.) As before, let us propagate $\Surs$ along the flow generated by $\xi^a$ and obtain a 3-manifold $\mathfrak{H}$, naturally foliated by 2-spheres $\Surs_\lambda$. Then $W_p$ is mapped to $W_{p,\,\lambda}$ that spans the tangent space of $\mathfrak{H}$. Since the flow is an isometry, if $W_p$ is initially timelike (respectively, null), then $W_{p,\,\lambda}$ would remain timelike (respectively, null). Next, let us consider the surfaces $\Surs_\lambda$ generated by the flow along $\xi^a$.  Let us denote by $\mathcal{C}$ the closed subset of $\Surs$ where $W_p$ is null,\, by $\mathcal{O} = \Surs\setminus \mathcal{C}$ the open subset where $W_p$ is timelike, and let us denote their images under the flow generated by $\xi^a$ by  $\mathcal{C}_\lambda$ and $\mathcal{O}_\lambda$, respectively; see Fig.~\ref{fig:fig1}. (For simplicity, we have depicted the situation in which $\mathcal{C}$ and $\mathcal{O}$ are simply connected. In the more general case, one just repeats the argument for each connected component.)
\begin{figure}
  \includegraphics[width=0.6\columnwidth]{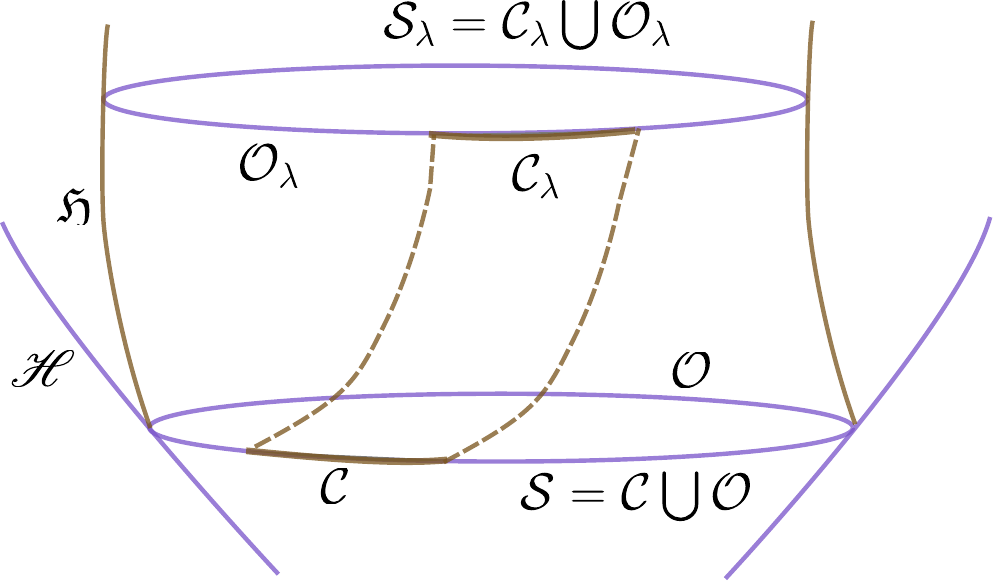}
  \caption{\footnotesize{Causal Killing field $\xi^a$: The dynamical horizon $\DHS$, and the 3-surface $\mathfrak{H}$ generated by the flow of $\xi^a$. A priori, $\xi^a$ could be null on a closed subset $\mathcal{C}$ of a MTS $\Surs$ in $\DHS$, and is timelike on its complement $\mathcal{O}$.}}
  \label{fig:fig1}
\end{figure}

Since $\Surs_\lambda = \mathcal{O}_\lambda \cup C_\lambda$,\, setting $\mathcal{O}_{\mathfrak{H}} := \bigcup_\lambda\mathcal{O}_\lambda$ and $C_{\mathfrak{H}} := \bigcup_\lambda C_\lambda$, we have $\mathfrak{H} = \mathcal{O}_{\mathfrak{H}} \cup C_{\mathfrak{H}}$. 
By construction, the flow generated by $\xi^a$ also preserves $\mathcal{O}_{\mathfrak{H}}$, as well as $\mathcal{C}_{\mathfrak{H}}$. Next, as before, let $V^a = f_1k^a + f_2\uk^a$ be the vector field on $\mathfrak{H}$ that is everywhere normal to the MTSs $\Surs_\lambda$, with affine parameter $\lambda$, that preserves the foliation of $\mathfrak{H}$ by  MTSs $\Surs_\lambda$. Then:\\ (i) $V^a$ is timelike on the open portion $\mathcal{O}_{\mathfrak{H}}$ of $\mathfrak{H}$, and null on the closed, null portion $\mathfrak{C}_{\mathfrak{H}}$; and,\\ (ii) The flow generated by $V^a$ also preserves $\mathcal{O}_{\mathfrak{H}}$, as well as $\mathcal{C}_{\mathfrak{H}}$.\\ 
Therefore we can use the same argument as in Proposition \ref{prop1}, but now using $\mathcal{O}_{\mathfrak{H}}$ --where $\xi^a$ is timelike-- in place of the full $\mathfrak{H}$. Again, on the one hand, the areas of $\mathcal{O}_\lambda$ must all be the same since they are isometrically related. On the other hand, the areas must be different because under the flow generated by $V^a$ the area either increases or decreases monotonically. Thus, we arrive at a contradiction unless $\mathcal{O}$ is empty, i.e., unless the 3-flats $W_p$ are everywhere null on $\Surs$ so that $C_{\mathfrak{H}} = \mathfrak{H}$.\, Thus $\mathfrak{H}$ must a null 3-manifold with $V^a$ as its null normal, foliated by MTSs $\Surs_\lambda$. 
\footnote{In particular, the causal Killing vector field $\xi^a$ must be null everywhere on any MTS $S$ of the DHS $\DHS$ under consideration.}

Let us investigate this case further. Since the null normal $V^a$ to $\mathfrak{H}$ is orthogonal to every MTS $\Surs_\lambda$, it must be proportional either to $k^a$ or $\uk^a$. If it were proportional to $\uk^a$, we would again have a contradiction involving areas of $\Surs_\lambda$ because the expansion $\Theta_{(\uk)}$ is everywhere negative or everywhere positive. Therefore, $V^a$ must be proportional to $k^a$ (for which the contradiction is avoided because $\Theta_{(k)} =0$). Thus $k^a$ is a null normal to $\mathfrak{H}$. We can use the freedom to rescale $k^a$ by a positive function to choose it to be a geodesic null normal. Then, by Raychaudhuri equation for $k^a$ we have
  \begin{equation} \label{Ray}
    \Lie_k\Theta_{(k)} = -\sigma_{ab}^{(k)}\sigma^{ab}_{(k)} - R_{ab}k^ak^b = -\E\,.
  \end{equation}
Since $\Theta_{(k)} =0 $ on every $\Surs_\lambda$, we would have $\E =0$ everywhere on $\mathfrak{H}$, and in particular on the MTS $\Surs$ which it shares with the given DHS $\DHS$. But by assumption 3 of Definition 2 of DHS, $\E$ cannot vanish identically on $\Surs$. Thus, again we would have a contradiction. Therefore we conclude that a space-time region $\scrR$ cannot admit causal Killing field that is nowhere vanishing on a DHS. 

Finally, let us consider the remaining case in which $\xi^a$ vanishes somewhere on $\DHS$. If it is nowhere vanishing \emph{on any one} MTS $\Surs$, by the argument given above, we would have a contradiction. Also, using the fact that a Killing vector field is completely characterized by its `Killing data' $(\xi^a, \nabla_a \xi_b)_p$ at any point, and that $\nabla_a \xi_b$ is anti-symmetric, it follows that $\xi^a$ cannot vanish on any open set of a 3-dimensional submanifold of $\scrM$ \cite{Ashtekar1978}. These two properties of $\xi^a$ imply that the subset $\mathcal{C}_\circ$ of $\Surs$ on which $\xi^a$ vanishes can be at most one dimensional, for any MTS $\Surs$. We can apply the above argument to the complement $\mathcal{O}_\circ$ of $\mathcal{C}_\circ$ to conclude that $W_p$ must be (3-dimensional and) null, so that the image of $\mathcal{O}_\circ$ under the flow generated by $\xi^a$ constitutes a null 3-manifold on which $\E$ vanishes everywhere. Therefore $\E$ must vanish on every MTS $\Surs$ except possibly on a one or zero dimensional subset thereof. Then by continuity $\E=0$ on every $\Surs$, which contradicts the assumption that $\Surs$ is an MTS in a DHS. Thus, even if one allows for the possibility for the Killing field $\xi^a$ to vanish on $\DHS$, we arrive at a contradiction.

Therefore, we have:
\begin{prop}
\label{prop2}
If a region $\scrR$ of $\scrM$ admits a causal Killing vector field, then it cannot contain an entire MTS $\Surs$ of a DHS $\DHS$. In particular, then it cannot contain $\DHS$.

\end{prop}

\section{Spacelike Killing Vectors}
\label{s4}

For causal Killing fields $\xi^a$ considered in section \ref{s3} we could conclude that, on any MTS $\Surs$, the subspace $W_p$ of $T_p$ (spanned by $\xi^a$ and tangent vectors of $\Surs$) is either timelike or null. For spacelike Killing field $\xi^a$, on the other hand, we have more possibilities. First, $\xi^a$ may be tangential to $\Surs$ everywhere, in which case $W_p$ would be spacelike and 2-dimensional. This simple possibility occurs, in particular,  if the DHS is axisymmetric or spherically symmetric. By contrast, the case when (the spacelike Killing field) $\xi^a$ is allowed to be transverse (at least on a portion of a MTS $\Surs$) is more involved because it branches into several sub-cases. For, on the (open) portion of an $\Surs$ on which $\xi^a$ is transverse, $W_p$ would be 3-dimensional but would not be restricted to be timelike or null as before; it could also be spacelike. Also, since $\xi^a$ could be transverse only on a portion of $\Surs$, the dimensionality of $W_p$  can also vary from point to point. Consequently we now have several sub-cases to investigate.
            
Let us begin with the simplest case: $W_p$ is 3-dimensional and is timelike, null or spacelike \emph{everywhere} on $\Surs$. We can rule out this possibility simply by adapting the arguments used in \ref{prop1} and \ref{prop2}. For completeness we include a sketch of the proof. Consider as before the 3-surface $\mathfrak{H}$ obtained by transporting the initial MTS $\Surs$ along the flow generated by $\xi^a$. $\mathfrak{H}$ is foliated by MTSs $\Surs_\lambda$ all of which must have the same area because $\xi^a$ is a Killing field. Furthermore, the signature of $W_p^{(\lambda)}$ is preserved along this flow.  Thus, if $W_p$ were initially spacelike/timelike/null, then so would be $\mathfrak{H}$. In the spacelike and timelike cases, $\mathfrak{H}$ is a DHS. Therefore the area of the $\Surs_\lambda$ must be monotonic, which contradicts our conclusion that  areas are the same. In the case when $\mathfrak{H}$ is null, we use the Raychaudhuri equation, along with condition 3 in Definition 2 (of a DHS), to arrive at a contradiction.

Complications arise when the signature of $W_p$ varies over $\Surs$, even when $W_p$ is 3-dimensional everywhere on $\Surs$.  A possible configuration is shown in Fig.~\ref{fig:fig2}.
\begin{figure}
  \includegraphics[width=0.6\columnwidth]{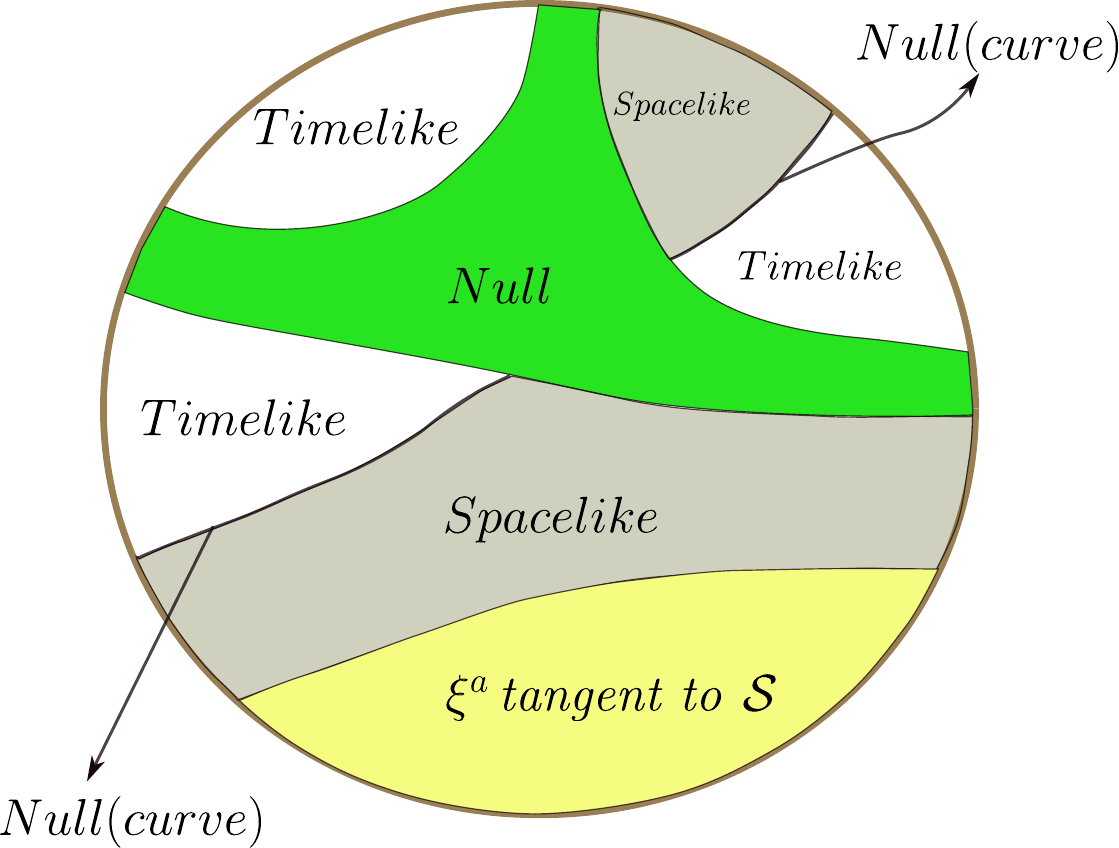}
\caption{\footnotesize{Spacelike Killing field $\xi^a$: An a priori possible configuration showing varying signature of $W_p$ on $\mathcal{S}$. $\xi^a$ is tangential on a portion of $\Surs$ where $W_p$ is 2-dimensional and spacelike. On portions to which $\xi^a$ is transverse, $W_p$ is 3-dimensional and can be timelike, null or spacelike.}}
  \label{fig:fig2}
\end{figure}
As before, denote by $\mathcal{C}$ the closed subset of $\Surs$ on which $W_p$ is null. The complement $\mathcal{O}$ may be either timelike or spacelike.  Moreover, if a connected component of $\mathcal{C}$ 
separates $\Surs$ into disconnected components, then connected components of $\mathcal{O}$ can have different signatures (as in Fig.~\ref{fig:fig2}).

Let us begin with the case when $\mathcal{C}$ consists of a single simply connected component. Then its complement $\mathcal{O}$ is also simply connected and either timelike or spacelike. Let us use the same notation as in Proposition \ref{prop2}. Then arguments used in Proposition \ref{prop2} imply that the union $\mathcal{C}_{\mathfrak{H}}$ of the images of $\mathcal{C}$ under the flow generated by $\xi^a$ is a 3-dimensional null sub-manifold that lies in $\mathfrak{H}$. By construction, $\mathcal{C}_{\mathfrak{H}}$ is preserved by the flow of $\xi^a$. On $\mathcal{C}_\mathfrak{H}$, the normal $V^a$ to the MTSs $\Surs_\lambda$ is null and it is tangential to the boundary $\partial \mathcal{C}_{\mathfrak{H}}$. Thus, as in Proposition \ref{prop2}, $\mathcal{C}_{\mathfrak{H}}$ and its complement $\mathcal{O}_{\mathfrak{H}}$ are preserved also by the flow generated by $V^a$. We can now repeat the previous argument, but using $\mathcal{O}_{\mathfrak{H}}$ in place of $\mathfrak{H}$. On the one hand, under the flow generated by $\xi^a$, the areas of $\mathcal{O}_\lambda$ must be the same as the area of $\mathcal{O} = \mathcal{O}_{\lambda=0}$.  On the other hand, they must differ because they are connected by the flow generated by the vector field $V^a$ normal to $\Surs_\lambda$ as well. Thus we arrive at a contradiction. Hence there cannot exist a spacelike Killing field $\xi^a$ that is everywhere transverse to a MTS $\Surs$ such that the the 3-flats $W_p$ it spans together with the tangent vectors to $\Surs$ are null on a single connected component of $\Surs$.  

If $\mathcal{C}$ is not simply connected, there would be several disconnected components of $\mathcal{O}$ (as in Fig.~\ref{fig:fig2}). These could be either timelike or spacelike. The area of the timelike portion decreases while the area of the spacelike portion will increase. Thus, the \emph{total} area of $\mathcal{O}_\lambda$ may remain unchanged. At first this may seem like an obstacle. However, we can still apply the above argument to a single connected component on which $W_p$ is either exclusively timelike or exclusively spacelike.  Arguing as above, we again arrive at a contradiction involving the area of this patch under the flow generated by $\xi^a$. Thus, the tangent space $W_p$ --spanned by the Killing field $\xi^a$ together with the the 2-d tangent space $T_\Surs$-- cannot be 3-dimensional everywhere on any MTS $\Surs$ of $\DHS$.   

Finally, the Killing vector field $\xi^a$ may be tangential on a closed subset $\t{\mathcal{C}}$ of every  MTS $\Surs$, and transverse elsewhere. In this case, $W_p$ would be 2-dimensional on $\t{\mathcal{C}}$ and 3-dimensional on its complement $\t{\mathcal{O}}$ (and $\t{\mathcal{C}}$ and  $\t{\mathcal{O}}$ are left invariant also by flows generated by $\xi^a$ as well as $V^a$). Now we can focus just on each connected component of $\t{\mathcal{O}}$ which would have a definite signature (i.e. would be null, timelike or spacelike everywhere). Repeating the argument used in Proposition \ref{prop2} on each component,  we would arrive at a contradiction unless the 3-manifold $\t{\mathcal{O}}_{\mathfrak{H}} = \bigcup_\lambda\, \t{\mathcal{O}}_\lambda$ is null everywhere. In that case, as in Proposition \ref{prop2} we can conclude that $\E$ must vanish on each $\t{\mathcal{O}}$. But now we do not have a contradiction because $\E$ only has to be non-vanishing \emph{somewhere} on $\Surs$ and it could be non-zero on the complement $\t{\mathcal{C}}$ of $\t{\mathcal{O}}$. Therefore an additional condition is needed. There are two alternatives.

First, let us restrict ourselves to DHSs on which $\E$ does not vanish on an open set of any MTS $\Surs$ of $\DHS$. Then as in proposition \ref{prop2} we \emph{would} have a contradiction and we can conclude that $\xi^a$ must be tangential to every MTS of $\DHS$. 

A second alternative is to assume that the null energy condition holds everywhere on $\DHS$. Then, using the Raychaudhuri equation (\ref{Ray}), it follows that $\sigma_{ab}^{(k)} =0$ on every MTS $\Surs$ which, together with the fact that $\Theta_{(k)} =0$ implies that $k^a$ is a Killing field for the intrinsic metric on the 3-manifold $\t{\mathcal{O}}_{\mathfrak{H}}$. Next, note that the restriction of $\xi^a$ to  $\t{\mathcal{O}}_{\mathfrak{H}}$ has the form $\xi^a = \t{\xi}^a + f k^a$, where $\t{\xi}^a$ is tangential to $\mathcal{O}_\lambda$ and the function $f$ vanishes on the boundary %$\partial\,\t{\mathcal{O}}_{\mathfrak{H}}$ 
of $\t{\mathcal{O}}_{\mathfrak{H}}$. Therefore on the MTS $\Surs$ of $\DHS$ we began with, $\t\xi^a$ is a smooth continuation of $\xi^a$ that is a Killing field of the metric $\t{q}_{ab}$ on $\Surs$. This holds for every MTS $\Surs$ on $\DHS$. For notational simplicity, let us denote this smooth vector field on $\DHS$ also by $\t\xi^a$. Thus, even when the (spacelike) Killing field $\xi^a$ fails to be tangential everywhere to MTSs $\Surs$, its projection $\t\xi^a$ provides us a Killing field on every MTS $\Surs$.

Collecting these results and those of Propositions \ref{prop1} and \ref{prop2}, we arrive at our final 
conclusion: 
\begin{prop}
  \label{prop3}
Let $(\scrM,\,g_{ab})$ be a space-time that admits a Killing field in a region $\scrR$ that contains a DHS $\DHS$. Then $\xi^a$ must be spacelike on $\DHS$ and it cannot be everywhere transverse to \emph{any} MTS $\Surs$ of $\DHS$. If it is everywhere tangential to $\DHS$, then it must be tangential to each MTS $\Surs$ of $\DHS$. Finally,\\ 
(i) if $\E$ does not vanish on any open subset of an MTS in $\DHS$, then $\xi^a$ must be everywhere tangential to that MTS;\\
(ii) If $\E$ does vanish on an open subset $\t{\mathcal{O}}$ of an MTS $\Surs$ in $\DHS$, then $\xi^a$ could be transverse to $\Surs$ on $\t{\mathcal{O}}$. However if the null energy condition holds on $\Surs$, then the projection $\t\xi^a$ of $\xi^a$ in to $\Surs$ is also a Killing field on $\Surs$.\\
In both cases, the intrinsic metric on each MTS $\Surs$ in $\DHS$ is equipped with a Killing field. 

\end{prop}

\textbf{Remarks:} 
\vskip0.1cm

1. Since each MTS $\Surs$ has 2-sphere topology, the induced metric $\t{q}_{ab}$ on $\Surs$ is conformally related to a unit, round, 2-sphere metric $\mathring{\t{q}}_{ab}$. There is a 3-parameter family of such round metrics that results from the 3-parameter family of `boosts' that maps unit, round 2-sphere metrics to one another. These considerations imply that a Killing field of the metrics $\t{q}_{ab}$ on MTSs $\Surs$ is a rotational Killing field of some round metric $\mathring{\t{q}}_{ab}$. (For details, see Appendix A.) Thus there are only two possibilities for symmetries induced on a DHS by space-time Killing fields $\xi^a$: Each MTS is either axisymmetric or spherically symmetric.
\vskip0.1cm

2. Thus, if the space-time metric admits a Killing field $\xi^a$ in a neighborhood of a DHS $\DHS$, then the metric on each MTS $\Surs$ is axisymmetric, whence one can define multipoles on $\DHS$ following the procedure summarized in section \ref{s2.2}.

\goodbreak
\section{Rotational Killing  fields and spin multipoles}
\label{s5}

In the last two sections we explored the constraints on admissible Killing fields $\xi^a$ in a neighborhood of a DHS $\DHS$. In this section we consider the complementary issue: restrictions on the geometry of $\DHS$ that are imposed by the presence a Killing field $\xi^a$ in its neighborhood. Results of sections \ref{s3} and \ref{s4} lead us to focus on DHSs in which each MTS is axisymmetric.

Recall from section \ref{s2.2} that the shape and the spin of any $\Surs$ in a $\DHS$ are determined, respectively, by the scalar curvature $\t\R$ and the (exterior derivative of the) rotational 1-form $\t\omega_a$, which in turn are determined by the shape and spin multipoles $\texttt{I}_{\ell,m}$ and 
$\texttt{S}_{\ell,m}$. The first and obvious restriction imposed by the presence of a rotational Killing field $\varphi^a$ is that all multipoles with $m\not=0$ vanish.
\footnote{Sometimes there is an additional discrete (reflection) symmetry $\zeta \to 1-\zeta$ (or $\theta \to \pi -\theta$), then all $\texttt{I}_\ell$ vanish for odd $\ell$, and all $\texttt{S}_\ell$ vanish for even $\ell$.}
A second restriction refers to the angular momentum dipole moment $\texttt{S}_{\ell=1}\equiv \texttt{S}_{1,0}$ of any MTS $\Surs$: It equals the Komar integral evaluated on it. Therefore it's change is dictated entirely by the flux of the matter angular momentum $T_{ab}\varphi^b$ across $\DHS$; it is conserved if there is no matter flux, even when there is flux of gravitational radiation across $\DHS$. 
In what follows, we will discuss a more subtle and rather surprising restriction on the spin multipoles.

In the axisymmetric case, the expression (\ref{dimensionless2}) of spin multipoles becomes:
\be \texttt{S}_{\ell} [\Surs] := \textstyle{\f{1}{2}}\, \oint_{\Surs} \big[ \,\t\epsilon^{ab} \t{D}_a \t\omega_b\,\big]\, Y_{\ell, 0} (\theta,\varphi) \, \rmd^2 \t{V}\, .\ee
Integrating by parts, expanding  $\t\epsilon^{ab}$, and using the expression (\ref{tomega}) of the rotational 1-form $\t\omega_a$ in terms of the extrinsic curvature $K_{ab}$ of $\DHS$ and the unit normal $\h{r}^a$ to the MTS $\Surs$ within $\DHS$, we have
\ba \label{main}\texttt{S}_\ell [\Surs] &=& -\f{1}{2R^2} \oint_{\Surs} (\varphi^a \t\omega_a)\,\partial_\zeta Y_{\ell}(\zeta)\, \rmd^2 \t{V} \nonumber\\
&=&  -\f{1}{2R^2} \oint_{\Surs} (\varphi^a\, K_{ab}\,\h{r}^b)\, \partial_\zeta Y_{\ell}(\zeta) \rmd^2 \t{V} \nonumber\\
&=&  -\f{1}{2R^2} \oint_{\Surs} (\varphi^a \nabla_b \h\tau_a)\, \h{r}^b\,  \partial_\zeta Y_{\ell}(\zeta) \rmd^2 \t{V} \nonumber\\
&=& \,\,\,\, {\f{1}{2R^2}} \oint_{\Surs} (\h\tau^a \nabla_b \varphi_a) \h{r}^b\,  \partial_\zeta Y_{\ell}(\zeta) \rmd^2 \t{V}\, . \ea
Finally, let us use the fact that the derivative of any Killing field $\varphi^a$ on a space-time $(\scrM,\, g_{ab})$ can be expressed in terms of its norm $\mathfrak{n} := g_{ab} \varphi^a \varphi^b$ and twist $\mathfrak{t}_a := \epsilon_{abcd}\, (\nabla^b \varphi^c)\, \varphi^d$:
\be \label{dphi}\nabla_a \varphi_b =  \f{1}{\mathfrak{n}}\, \big[ \varphi_{[b} \nabla_{a]} \mathfrak{n})\,  -\, \f{1}{2}\,\epsilon_{abcd}\, \mathfrak{t}^c \varphi^d \big] \ee
Using this expression  of $\nabla_b \varphi_a$ in the last step of (\ref{main}) and simplifying, we have:
\be \label{final}\texttt{S}_\ell [\Surs] = \f{1}{4} \oint_{\Surs} \f{1}{\mathfrak{n}}\, \big(\mathfrak{t}^a \partial_a Y_\ell(\zeta)\big)\, \rmd^2 \t{V}. \ee
Thus, spin moments $S_\ell [\Surs]$ are all completely determined by the twist $\mathfrak{t}^a$ of $\varphi^a$. In particular they \emph{all vanish} if $\mathfrak{t}^a$ vanishes on $\Surs$, i.e., if $\varphi^a$ is hypersurface orthogonal there.

This is rather surprising because in a general axisymmetric space-time, we have a constraint only on the angular momentum dipole moment $\texttt{S}_{\ell=1}[\Surs]\,\, (\equiv\, {J}_\varphi [\Surs])$ : as noted above, it equals the Komar integral and is conserved if vacuum equations hold in a neighborhood of $\DHS$. Not only are higher spin-multipoles non-zero, but they are generically \emph{time-dependent} even when vacuum equations hold! (See Eqs. (3.13) and (3.14) of \cite{Ashtekar:2013qta} for expressions of the flux.) If, on the other hand, $\varphi^a$ is hypersurface orthogonal at $\DHS$, then not only the fluxes, but the spin multipoles themselves vanish. An example in which this situation occurs is provided by the head-on collision of non-spinning black holes investigated in detail in \cite{Pook-Kolb:2018igu,PhysRevLett.123.171102}. At first it was rather puzzling as to why this high precision numerical simulation yielded a DHS with vanishing spin moments, even though it is not spherically symmetric. Of course, the shape multipoles $\texttt{I}_\ell$ of such a DHS are quite different from those of a spherically symmetric DHS: while all $\texttt{I}_\ell$ vanish except for $\ell=0$ on spherical DHSs, not only are they non-zero on axisymmetric DHSs with hypersurface orthogonal $\varphi^a$, but they have interesting and highly non-trivial dynamics.

\section{Discussion}
\label{s6}
 
In this note we have summarized the interplay between DHSs and space-time isometries in their neighborhood. DHSs can be thought of as `world tubes' of MTSs. In numerical simulations they are often obtained by first introducing a space-time foliation, and then locating MTSs on each (partial) Cauchy slice. 
\footnote{Typically, one locates the apparent horizon on each slice but this strategy fails to find DHSs corresponding to the progenitors after a merger. One needs advanced horizon finders that can locate MTSs also inside apparent horizons \cite{pook_kolb_daniel_2019_2591105}.}
However, geometrically, a DHS is simply a 3-dimensional submanifold of space-time in its own right without reference to any foliation. Therefore in this note we considered Killing fields $\xi^a$ in neighborhoods $\mathcal{R}$ of DHSs $\DHS$, without reference to any foliation of space-time. We found that the very presence of $\DHS$ leads to strong restrictions on permissible $\xi^a$ in $\mathcal{R}$.  There are several results in the literature on this interplay. However, they involve additional assumptions, e.g., the requirement that the Killing field be strictly stationary \cite{Mars:2003ud}; or that the DHS be spacelike, with strictly negative $\Theta_{\uk}$ and nowhere vanishing $\E = \sigma_{ab}^{(k)} \sigma^{ab}_{({k)}}  + R_{ab} k^a k^b$ \cite{Ashtekar:2005ez}; or that the MTS under consideration is strictly stable in certain directions \cite{Andersson:2007fh}; or the reasoning is tied to a space-time foliation \cite{Carrasco:2009sa,Booth_2024}. In this note we focused on DHSs without reference to space-time foliations, and generalized these results by dropping additional assumptions, and including possibilities for the behavior of $\xi^a$ at $\DHS$ that were not considered before. 

Our main result is contained in Proposition \ref{prop3} of section \ref{s4}. First, it rules out the presence Killing vector $\xi^a$ that is Causal on $\DHS$. On the one hand this is what one would intuitively expect because the intrinsic geometry on any $\DHS$ is dynamical. However, the result is not obvious because, e.g., spacelike world tubes of MTSs with growing areas can exist in presence of a null Killing field \cite{Senovilla:2003tw}. It is the specific conditions in Definition 2 of DHSs that forbid the existence of such Killing fields in a neighborhood of $\DHS$. Thus, the absence of causal Killing fields provides some confidence that the conditions in Definition 3 are well tuned to capture the intuitive connotation one has in mind when one speaks of DHSs. 

Second, Proposition \ref{prop3} also puts strong restrictions on possible spacelike Killing fields $\xi^a$ in a neighborhood $\mathcal{R}$ of $\DHS$:\\ (i) $\xi^a$ cannot be transverse to \emph{any} entire MTS $\Surs$ in $\DHS$;\,\\ (ii) if it is tangential to $\DHS$, it must be tangential to every $\Surs$ in $\DHS$;\\ (iii) if it is tangential on a closed subset $\mathcal{C}$ of a MTS $\Surs$ and transverse on its complement $\mathcal{O}$, then the 3-flat $W_p$ spanned by $\xi^a$ and the tangent space $T_\Surs$ must be null everywhere on $\mathcal{O}$, and $\E$ must vanish on $\mathcal{O}$. Clearly,  this possibility cannot be realized if $\E$ does not vanish on open sets of any $\Surs$, as is often assumed;\, and,\\ (iv) Even when possibility (iii) is realized, the projection $\t{\xi}^a$ of $\xi^a$ on $\Surs$ is a rotational Killing field on $\Surs$ if the null energy condition holds. This is in particular the case in vacuum space-times considered in black hole coalescence and scattering.\\ Thus, if there is a Killing field $\xi^a$ in a neighborhood $\mathcal{R}$ of $\DHS$, then generically each MTS $\Surs$ of $\DHS$ is axisymmetric. 

\emph{Hence, generically MTSs $\Surs$ of a DHS in space-times admitting Killing fields are either axisymmetric or spherically symmetric.} No other symmetries are permissible. By contrast, non-expanding horizons (NEHs) also allow space-time isometries that are null. Indeed the familiar Kerr event horizons are particular cases of NEHs. These symmetries are not permissible on DHSs because their geometries are dynamical. 

These results have an interesting implication for DHSs that arise in numerical simulations where DHSs  are generated by time evolution of MTSs. Suppose the space-time under consideration is axisymmetric and does admit a DHS $\DHS$ generated, e.g., by finding a MTS on each leaf of an axisymmetric foliation, i.e., one in which the rotational Killing field $\varphi^a$ is tangential to each leaf. Now consider a different foliation such that $\varphi^a$ is \emph{everywhere} transverse to each leaf. If such a foliation were to generate a DHS $\DHS^\prime$ then on each MTS $\Surs^\prime$ of  $\DHS^\prime$, the Killing field would be everywhere transversal. But this is not possible by result (ii) above. Thus, to generate a DHS, one must use a foliation to which $\varphi^a$ is at least partially tangential.  
{\footnote{Furthermore, the result by Booth et al \cite{Booth_2024} mentioned in section \ref{s2.3} implies that even when  $\varphi^a$ is tangential to any one leaf, if it fails to be tangential to the MTS on that leaf, then that MTS would not be stable. Therefore, in this case we are not guaranteed that the time evolution of that MTS would yield a DHS (although it could, since MTSs on a DHS need not be stable).} }
If  a smoothly connected world-tube of MTSs were found using a foliation to which $\xi^a$ is nowhere tangential, it would not constitute a DHS; at least one of the conditions in Definition 2 must is violated. In particular, the signature of the metric and/or the sign of $\Theta_{(\uk)}$ may not be constant on the world-tube of MTSs. 

Interestingly, asymptotic flatness at null infinity, $\scri$, also imposes strong restrictions on isometries that the space-time can admit near $\scri$  \cite{doi:10.1063/1.524467}. For the appropriate comparison with DHSs we need to allow the possibility that $\scri$ is incomplete, and restrict ourselves to situations in which there is  gravitational radiation and/or matter flux across $\scri$.  In that case, axisymmetry and spherical symmetry are of course permissible. But there is also a possibility of a 1-dimensional isometry group with a boost symmetry, and a 2-dimensional isometry group with a boost-rotational symmetry \cite{BicakSchmidt1989}. Indeed, this second possibility is realized in the c-metric \cite{Ashtekar:1981} which admits a $\scri$ with $S^2\times R$ topology with non-vanishing Bondi news. Since it represents two accelerating black holes and there is gravitational radiation, why is this symmetry not allowed in the DHS framework? This is because in this case the solutions have `nodal singularities' on 2-spheres whence they do not admit smooth DHSs. 
\footnote{At $\scrip$ these nodal singularities appear only on two `half, generators, to the future or past of the so-called `bullet holes' at which the black holes pierce $\scrip$. Therefore we do have a $S^2\times R$ portion of $\scrip$ without nodal singularities.}
As specified at the end of section \ref{s1}, our results assume smooth geometries which preclude such nodal singularities (also, see  Appendix A).

Finally let us summarize results of section \ref{s5} on constraints on the spin structure of DHSs that are imposed by the presence of a Killing vector field. First, because in presence of an isometry the intrinsic geometry of every MTS $\Surs$ of $\DHS$ is axisymmetric, we can use this axial Killing field $\t\varphi^a$ to define multipoles following the procedure outlined in section \ref{s2.2}. The presence of $\t\varphi^a$ constrains the multipoles. The first set of constraints is rather obvious:\\ (i) The shape multipoles $\texttt{I}_{\ell, m}$ and the spin multiples $\texttt{S}_{\ell, m}$ all vanish for $m\not=0$; and, \\(ii) The spin dipole moment agrees with the Komar integral, whence its flux is governed entirely by matter field.\\ But surprisingly, \emph{all} spin-multipoles vanish identically if $\t\varphi^a$ is hypersurface orthogonal! \emph{Thus, the spin structure of such a DHS is indistinguishable from that of a spherically symmetric DHS!} It would not have been surprising at all if just the spin dipole vanished in this case because it equals the Komar integral and it is well known that it vanishes if the rotational Killing field is hypersurface orthogonal. That this must be the case also for \emph{all} higher spin multipoles was not foreseen; it would have been less surprising if, e.g., \emph{only the fluxes} of spin multipoles vanished, i.e., the moments were conserved. In the process of establishing this result, we also found an explicit and expression (\ref{final}) for spin multipoles $\texttt{S}_{\ell}$ in terms of the norm and the twist of the Killing field for general axisymmetric DHSs. This expression may be useful in numerical simulations of axisymmetric collapse and in phases of binary coalescence where DHSs are found to be axisymmetric within numerical accuracy.
\goodbreak

\section*{Acknowledgments} 
This work was supported in part by the Atherton and Eberly funds of Penn State and the Distinguished Visiting Research Chair Program of the Perimeter Institute.

\begin{appendix}
\section{Killing fields that a metric can admit on $S^2$.}
\label{a1}

Consider a smooth metric $\t{q}_{ab}$ on a 2-manifold that is topologically $S^2$. Then, as is well  known, $\t{q}_{ab}$ is conformally related to a unit, round, 2-sphere metric $\qot_{ab}$:\,\, $\t{q}_{ab} = \t{\psi}^2 \qot_{ab}$. We can regard the round 2-sphere as residing in the Euclidean 3-space $(R^3, \qo_{ab}) $, given by ${\rm x}^a {\rm x}_a = 1$, where ${\rm x}^a$ is the position vector in $R^3$ and, as usual, the indices are raised and lowered using the Euclidean 3-metric $\qo_{ab}$.
\footnote{In this Appendix, we work in $R^3$. In particular, abstract indices $a, b, \ldots$ refer to the tangent space of $R^3$. Tensor fields what are intrinsic to (i.e. tangential to) $S^2$ will carry a tilde. Thus, $\qot_{ab}$ is the metric induced on $S^2$ by the Euclidean metric $\qo_{ab}$ on $R^3$.}
We will denote the \emph{unit} radial vector field in $R^3$ by $\h{r}^a$. In this Appendix, (the abstract) latin indices $a,b,\ldots$ refer to the tangent space of $R^3$ and fields carrying a tilde are intrinsic to $S^2$, i.e., their contraction with $\h{r}^a$ vanish.

The round metric $\qot_{ab}$ is not unique: there is precisely a 3-parameter family of such unit round metrics that are themselves conformally related to one another:
\be \label{trans} \qot_{ab}^{\,\prime}\,  \= \,\psio^2\, \qot_{ab},\quad {\rm where} \quad \psio^{-1}\, \= \,{\alpha_\circ \,+\, \alpha_a \hr^a}, \quad{\rm with} \quad \alpha_\circ - \alpha_a \alpha^a = 1\ee
where $\=$ denotes restriction to the unit $S^2$ in $R^3$, $\alpha_\circ$ is a constant and $\alpha^a$ a constant vector field on $R^3$. In spherical polar coordinates on $R^3$ we have:
\be \psio^{-1}\, \=\, {\alpha_\circ + \alpha_1\, \sin\theta \cos\phi + \alpha_2\, \sin\theta \sin\phi+ \alpha_3\, \cos\theta}\,.\ee
In this Appendix we show that if $\t{q}_{ab}$ admits a Killing field $\t\xi^a$, then there exists a unit round metric $\qot_{ab}^{\,\prime}$  for which $\t\xi^a$ is a Killing field; thus $\t\xi^a$ is a `rotation'. This is the assertion of Remark 1 at the end of section \ref{s4}.

Let us begin by noting a few facts about symmetries of round 2-sphere metrics $\qot_{ab}$. Since $(S^2,\, \qot_{ab})$ is a manifold of constant curvature, it admits six conformal Killing fields $\t\xi^a$, the maximum number permissible on a 2-manifold. They can be expressed as:
\be \label{xi} \t\xi^a \,\=\, \t{B}^a + \t{R}^{a}\, ,\qquad {\rm with} \qquad \t{B}^a\, \=\,  \qot^{ab}\Ko_b\quad {\rm and}\quad 
\t{R}^a\, \=\, \eot^{ab} \Lo_b \ee
where  $\Ko^a$ and $\Lo^a$ are \emph{constant} vector fields on $R^3$, while $\qot_{ab}$ and $\eot_{ab}$ are the intrinsic metric and the area 2-form on $S^2$, induced by the Euclidean metric $\qo_{ab}$:
\be \qot_{ab}\, \=\, \qo_{ab} -\hr_a \hr_b\qquad {\rm and}\qquad \eot_{ab}\, \=\,\, \eo_{mab}\hr^m\,. \ee
It is easy to check that $\t{R}^a$ is a Killing vector field of $\qot_{ab}$ on $S^2$,\, i.e. ${\Dot}_{(a}  \t{R}_{b)}\, \=\,0$,\, and $\t{B}^a$ is a conformal Killing field, i.e.,\, ${\Dot}_{(a} \t{B}_{b)\,} \= -\Ko^a\hr_a$, where  ${\Dot}$ is the intrinsic derivative operator on $S^2$ compatible with $\qot_{ab}$. Interestingly, $\t{B}_b$ is a gradient,\,  $\t{B}_a \= \Dot_a (K_c\hr^c)$,\, whence ${\Dot}_{[a} \t{B}_{b]}\,\, \=\, 0$. Thus, the derivative of $\t{R}_a$ is purely anti-symmetric, while that of $\t{B}_a$ is purely symmetric. In this sense, while $\t{R}^a$ is a Killing field, $\t{B}^a$ can be regarded as a `pure' conformal Killing field. To summarize, there is a neat division of the conformal Killing fields $\t\xi^a$  on $(S^2, \qot_{ab})$ into Killing fields $\t{R}^a$ (one for each constant vector fields $L^a$ in $R^3$) and `pure' conformal Killing fields $\t{B}^a$ (one for each constant vector field $K^a$ in $R^3$).  (This division persists on all n-manifolds of constant curvature.) On $S^2$, the Killing vector fields $\t{R}^a$ of $\qot_{ab}$ are rotations, with closed orbits and precisely two zeros. The `pure' conformal Killing fields have the interpretation of `boosts' --that descends from the action of the Lorentz group on the light cone of a point in Minkowski space. (On any 2-sphere cross-section of $\scri$ that is left invariant by the action of a Lorentz subgroup of the Bondi-Metzner Sachs group, rotations and boosts assume the form of our $\t{R}^a$ and $\t{B}^a$ in any Bondi conformal frame.)

Let us return to the general metric $\t{q}_{ab}$ on $S^2$. Since it is conformally related to $\qot_{ab}$, any Killing field $\t\xi^a$ of $\t{q}_{ab}$ is a conformal Killing field of $\qot_{ab}$. Therefore it is a linear combination of a boost $\t{B}^a$ and a rotation $\t{R}^a$ of $\qot_{ab}$ as in (\ref{xi}). Therefore, given a Killing field $\t\xi^a$ of $\t{q}_{ab}$, we have:
\be 0\,=\, \Lie_{\t\xi}\, \t{q}_{ab}\, \=\, \big(\Lie_{\t\xi} \ln \t{\psi} - 2 (\Ko\cdot \hr)\big)\, \t{q}_{ab}\,\ee
whence, 
\be \label{Lie} \Lie_{\t\xi}\, \ln \t{\psi}\, \= 2\, (\Ko\cdot \hr)\, \quad {\hbox{\rm everywhere on $S^2$}}.\ee
Since the conformal factor $\t{\psi}$ is smooth and non-zero everywhere on $\Surs$, it follows that at points on $S^2$ where the Killing vector $\t\xi^a$ vanishes, the right side must also vanish, i.e. that the
zeros of $\xi^a$ must lie on the the great circle on which $\Ko\cdot \hr =0$. Using the form (\ref{xi}) of $\t\xi^a$ one can show that this is possible if and only if  $\Ko\cdot \Lo =0$ and $\Ko\cdot \Ko < \Lo\cdot \Lo$. 

Let us then restrict ourselves to $\t\xi^a$ for which $\Ko^a$ and $\Lo^a$ satisfy this condition.
\footnote{For example, if we choose $K^a$ to point along $z$ axis then $\t{B}^a$ is a $z-t$ boost. In this case, $\Lo^a$ must lie in the $x-y$ plane so that $\t{R}^a$ is a rotation along an axis in that plane. These interpretations refer to the unit round metric  $\qot_{ab}$ used in the expression (\ref{xi}) of $\t\xi^a$.}
Then, the Killing vector $\t\xi^a$ of the given metric $\t{q}_{ab}$ is a linear combination of a rotation and a boost w.r.t. the reference round metric $\qot_{ab}$ used in the expression (\ref{xi}). But we know that there is a 3 parameter freedom in the choice of a unit round metric on $S^3$. Can we choose another round metric $\qot^{\,\prime}_{ab}$ of (\ref{trans}) such that $\xi^a$ is a pure rotation with respect to $\qot^\prime_{ab}$? Thanks to the conditions satisfied by $K^a$ and $L^a$, the answer is in the affirmative!  To obtain $\qot^{\,\prime}_{ab}$, we have to choose the constant $\alpha_\circ$ and the constant vector field $\alpha^a$ in (\ref{trans}) such that:
\be \mathring\epsilon^{abc}\, \alpha_a \,\Lo_b = \alpha_\circ \Ko^c\, .  \ee
This is possible precisely because $\Ko\cdot \Lo =0$ and $\Ko\cdot \Ko < \Lo\cdot \Lo$, and there is a one parameter freedom in the choice of the required $\alpha^a$. Thus, given a Killing field $\t\xi^a$ of $\t{q}_{ab}$, there exists a 1-parameter family of unit, round metric $\qot^{\,\prime}_{ab} = \psio^2\,\t{\psi}^{-2} \t{q}_{ab}$ such that $\xi^a$ is a pure rotation w.r.t. $\qot^{\,\prime}_{ab}$. 

\vskip0.1cm
\textbf{Remark:} 

As discussed in section \ref{s6}, the c-metric does admit a boost Killing field. On the other hand, conditions on $\Ko$ and $\Lo$ imply that a Killing field $\xi^a$ of the physical metric $\t{q}_{ab}$ cannot be a pure boost. It is instructive to see what goes wrong if we set $\t{xi}^a = \t{B}^a$. Then, an explicit calculation shows that the solution $\t\psi$ to  Eq. (\ref{Lie}) must diverge at the `poles' where which $\t{B}^a$ vanishes, violating the assumption that $\t{q}_{ab}$ is smooth. (In the c-metric these `poles' are the locations of nodal singularities of $\t{q}_{ab}$.) Thus, smoothness of the metric of DHSs prevents the boost to be a Killing field of $\t{q}_{ab}$. In particular, then, the 2-d (Abelian) boost-rotation subgroup of the conformal group of $\qot_{ab}$ cannot be part of the isometry group in a neighborhood $\mathcal{R}$ of a DHS. We only have two possible isometric groups in $\mathcal{R}$: either a single rotational Killing field (axisymmetry), or three (spherical symmetry).

\end{appendix}

\bibliography{symmetry}{}

\end{document}